\documentclass[runningheads]{llncs}
\usepackage{llncsdoc}
\usepackage{url}
\urldef{\mailjs}\path|jakub.szymanik@gmail.com|
\urldef{\mailrh}\path|dehaan@ac.tuwien.ac.at|

\usepackage{hyperref}

\usepackage{enumerate}
\usepackage[T1]{fontenc}

\usepackage{tikz}
\usepackage{amsfonts}
\usepackage[cmex10]{amsmath}
\usepackage{amssymb}

\usepackage{mathtools}

\usepackage{enumitem}
\setenumerate{leftmargin=2pc}
\newtheorem{assumption}[theorem]{Assumption}

\renewenvironment{proof}{\vspace{-2mm}\begin{pf}}{\qed\end{pf}}

\DeclareMathOperator{\fo}{FO}

\newcommand{\Q}{{\sf Q}}

\newcommand{\Br}{{\sf Br}}

\newcommand{\SB}{\{\,}%
\newcommand{\SM}{\;{:}\;}%
\newcommand{\SE}{\,\}}%
\newcommand{\SBs}{\{}%
\newcommand{\SEs}{\}}%

\renewcommand{\P}{\text{\normalfont P}}
\newcommand{\NP}{\text{\normalfont NP}}
\newcommand{\NPI}{\text{\normalfont NPI}}

 \newcommand{\FFF}{\mathcal{F}}

\newcommand{\QQQ}{\mathcal{Q}}

\newcommand{\mtext}[1]{\text{\normalfont #1}}

\newcommand{\kclique}[0]{\ensuremath{k\mtext{\sc -clique}}}

\author{Ronald de Haan \inst{1} and Jakub Szymanik \inst{2}}

\title{Characterizing Polynomial Ramsey Quantifiers}
\institute{Algorithms and Complexity Group, Vienna University of Technology \\\mailrh
\and Institute for Logic, Language and Computation, University of Amsterdam \\\mailjs}

\begin{document}

\maketitle

\begin{abstract}
Ramsey quantifiers are a natural object of study not only for logic and computer science, but also for the formal semantics of natural language. Restricting attention to finite models leads to the natural question whether all Ramsey quantifiers are either polynomial-time computable or NP-hard, and whether we can give a natural characterization of the polynomial-time computable quantifiers. In this paper, we first show that there exist intermediate Ramsey quantifiers and then we prove a dichotomy result for a large and natural class of Ramsey quantifiers, based on a reasonable and widely-believed complexity assumption. We show that the polynomial-time computable quantifiers in this class are exactly the constant-log-bounded Ramsey quantifiers.
\end{abstract}

\section{Motivations}
Traditionally, definability questions have been  a mathematical core of generalized quantifier theory. For example, over the years efforts have been directed to classify quantifier constructions with respect to their expressive power (see \cite{PW06} for an extensive overview). Another already classical feature of the theory is searching for linguistic and later computer-science applications. That is one of the reasons to often investigate quantifiers over finite models.
This leads naturally to questions about computational complexity (see \cite{Szymanik:2015kq} for an extensive overview). In previous work \cite{BG86,MW04,Sev06}, it has been observed that some generalized quantifier constructions, when assuming their branching interpretation, are NP-complete.\footnote{Sevenster has also proved a dichotomy theorem for independent-friendly quantifier prefixes that can capture branching quantification, namely they are either decidable in LOGSPACE or NP-hard \cite{Sevenster:2014uq}.}  
Following this line of research, Szymanik \cite{Szym09,Szy2011} searched for more natural classes of intractable generalized quantifiers. He found out that some natural language reciprocal sentences with quantified antecedents under the strong interpretation (see \cite {DKKMP98}) define NP-complete classes of finite models.\footnote{These results have interestingly also found empirical interpretations, see \cite{Schlotterbeck:2013rz,Thorne:2015kx}.} From a mathematical perspective these constructions correspond to Ramsey quantifiers. This leads to a natural mathematical question about characterization of Ramsey quantifiers with respect to their computational complexity.

In this technical paper we first  study some natural polynomial and NP-hard cases of Ramsey quantifiers. Next we ask whether all Ramsey quantifiers are either polynomial-time computable or NP-hard. We give a negative answer by showing that there exist intermediate Ramsey quantifiers -- i.e., Ramsey quantifiers that are neither polynomial-time computable nor NP-complete -- assuming a reasonable and widely-believed complexity assumption (the Exponential Time Hypothesis). This leads to another question:  is there a natural characterization of the polynomial Ramsey quantifiers? To positively answer that question we show that the Ramseyification of constant-log-bounded monadic quantifiers (based on polynomial-time computable threshold functions) results in polynomial time computable Ramsey quantifiers, while the Ramseyification of monadic quantifiers that are not constant-log-bounded are not polynomial-time solvable, assuming the Exponential Time Hypothesis.
That is, we give a dichotomy result where we can identify exactly which quantifiers lead to a polynomial-time solvable problem
-- this dichotomy result is based on the Exponential-Time Hypothesis, whereas similar dichotomy results in the literature are typically based on the complexity-theoretic assumption that~$\P \neq \NP$.
The notion of constant-log-boundedness is a version of the boundedness condition known from finite-model theory literature  \cite{Van97}, where the bound on the upper side is replaced by $c \log n$ for some constant~$c$, where $n$ is the size of the universe. As the property of boundedness plays an important role in definability theory of polyadic quantifiers \cite{HVW97}, we conclude by asking whether the new notion of constant-log-boundedness gives rise to some interesting descriptive results.


\section{Preliminaries}

\subsection{Generalized Quantifiers}
Generalized quantifiers might be defined as classes of models (see, e.g.,~\cite{PW06}). The formal definition is as follows:
\begin{definition}[\cite{Lin66}]
\label{defGQs}
Let \(t=(n_1, \ldots, n_k)\) be a $k$-tuple of positive integers. {\it A generalized quantifier} of type  
\(t\) is a class \(\Q\) of models of a vocabulary \(\tau_t=\{R_1, \ldots, R_k\}\), such that \(R_i\) is \(n_i\)-ary for \(1\le i\le k \), and \(\Q\) is closed under isomorphisms\index{isomorphism}, i.e. if $\mathbb{M}$ and $\mathbb{M'}$ are isomorphic, then \[(\mathbb{M} \in \Q \iff \mathbb{M'} \in \Q).\]
\end{definition}

Finite models can be encoded as finite strings over some vocabulary, hence, we can easily fit the notions into the descriptive complexity paradigm  (see e.g.~\cite{Imm98}):


\begin{definition}
By the {\it complexity of a quantifier $\Q$} we mean the computational complexity of the corresponding class of finite models, that is, the complexity of deciding whether a given finite model belongs to this class.
\end{definition}


For some interesting early results on the computational complexity of  various forms of quantification, see the work of Blass and Gurevich \cite{BG86}.

\subsection{Computational Complexity}\label{NPI}

We assume that the reader is familiar with the complexity classes
\P {} and NP (for an introduction to these classes, we refer
to textbooks on complexity theory, e.g.,~\cite{AroraBarak09}).

Problems in \NP{} that are neither in \P{} nor \NP{}-complete are called \NP{}-intermediate, and the class of such problems is called \NPI{}. Ladner \cite{Ladner:1975ty} proved the following seminal result:

\begin{theorem}[\cite{Ladner:1975ty}]
If $\P \neq \NP$, then \NPI{} is not empty.
\end{theorem}

Assuming $\P \neq \NP$, Ladner constructed an artificial \NPI{} problem.
Schaefer~\cite{Schaefer:1978rz} proved a dichotomy theorem for Boolean constraint satisfaction, thereby providing conditions under which classes of Boolean constraint satisfaction problems can not be in \NPI{}. It remains an interesting open question whether there are natural problems in \NPI{}~\cite{GKL07}.

\subsubsection{Asymptotic growth rates}

In this paper, we will use the concepts of ``big-O'',
``big-Omega,'' and ``little-o,'' that can be used to compare the asymptotic
growth rates of functions~$f : \mathbb{N} \rightarrow \mathbb{N}$
to each other.
These are defined as follows.

\begin{definition}
Let~$f,g : \mathbb{N} \rightarrow \mathbb{N}$ be computable functions.
Then~$f(n) \in O(g(n))$ if there exists some~$n_0 \in \mathbb{N}$
and some~$c \in \mathbb{N}$
such that for all~$n \in \mathbb{N}$ with~$n \geq n_0$ it holds that:
\[ f(n) \leq c g(n). \]
When~$f(n) \in O(g(n))$, we also say that~\emph{$f(n)$ is~$O(g(n))$}.
\end{definition}
Intuitively, if a function~$f(n)$ is~$O(g(n))$, it means that~$f(n)$ grows
at most as fast as~$g(n)$, when the values for~$n$ get large enough --
$f(n)$ grows asymptotically at most as fast as~$g(n)$.

\begin{definition}
Let~$f,g : \mathbb{N} \rightarrow \mathbb{N}$ be computable functions.
Then~$f(n) \in \Omega(g(n))$ if~$g(n) \in O(f(n))$.
When~$f(n) \in \Omega(g(n))$, we also say that~\emph{$f(n)$
is~$\Omega(g(n))$}.
\end{definition}
That is, ``big-Omega'' is the inverse of ``big-O.''
Intuitively, if a function~$f(n)$ is~$\Omega(g(n))$, it means that~$f(n)$
grows at least as fast as~$g(n)$, when the values for~$n$ get large enough --
$f(n)$ grows asymptotically at least as fast as~$g(n)$.

\begin{definition}[Definition 3.22 and Lemma 3.23 in \cite{Flum:2006fk}]
Let~$f,g : \mathbb{N} \rightarrow \mathbb{N}$ be computable functions.
Then~$f(n) \in o(g(n))$ 
if there is a computable function~$h$
such that for all~$\ell \geq 1$ and~$n \geq h(\ell)$,
we have:
\[ f(n) \leq \frac{g(n)}{\ell}. \]
Alternatively, the following definition is equivalent.
We have that~$f(n) \in o(g(n))$ if there exists~$n_0 \in \mathbb{N}$
and a computable function~$\iota : \mathbb{N} \rightarrow \mathbb{N}$
that is nondecreasing and unbounded such that for
all~$n \geq n_0$, it holds that~$f(n) \leq \frac{g(n)}{\iota(n)}$.
When~$f(n) \in o(g(n))$, we also say that~\emph{$f(n)$ is~$o(g(n))$}.
\end{definition}
Intuitively, if a function~$f(n)$ is~$o(g(n))$,
it means that~$g(n)$ grows faster than~$f(n)$,
when the values for~$n$ get large enough --
$g(n)$ grows asymptotically faster than~$f(n)$.

\subsubsection{The Exponential Time Hypothesis}\label{ETH}
The Exponential Time Hypothesis (ETH) says that $\textsc{3-SAT}$ (or any of several related NP-complete problems) cannot be solved in subexponential time in the worst case \cite{Impagliazzo:1999uq}. The ETH implies that $\P \neq \NP$. It also provides a way to obtain lower bounds on the running time of algorithms solving certain fundamental computational problems \cite{Lokshtanov:2011fk}.
Formally, the ETH can be stated as follows.

\begin{quote}
\textbf{Exponential Time Hypothesis:}\\
$\textsc{3-SAT}$ cannot be solved in time~$2^{o(n)}$,
where~$n$ denotes the number of variables in the input formula.
\end{quote}
\noindent The following result, that we will use to prove the existence
of intermediate Ramsey quantifiers (assuming the ETH)
is an example of a lower bound based on the ETH.
This result is about the problem $k$-\textsc{clique}, which involves finding cliques
(i.e., subgraphs that are complete graphs) of a certain size in a given graph.
For the problem $k$-\textsc{clique}, the input is a simple graph~$G = (V,E)$ and a
positive integer~$k$. The questions is whether~$G$ contains a clique
of size~$k$.

\begin{theorem}[\cite{Chen:2005fv}]
\label{thm:kclique-eth}
Assuming the ETH, there is no~$f(k)m^{o(k)}$ time
algorithm for $k$-\textsc{clique},
where~$m$ is the input size
and where~$f$ is a computable function.
\end{theorem}

\section{Ramsey Theory and Quantifiers}
\label{RamQua} \label{CompRam} 


Informally speaking the Finite Ramsey Theorem \cite{Ram29} states the following:%
\footnote{More precisely, the Finite Ramsey Theorem states that there exists a function~$R : \mathbb{N} \rightarrow \mathbb{N}$ such that for any coloring of the complete graph with~$R(n)$ vertices that colors the edges with two colors, there is a monochromatic subgraph of size~$n$.}

\

\textbf{The Finite Ramsey Theorem --- general schema} {\it When coloring the edges of a sufficiently large complete finite graph, one will find a homogeneous subset of arbitrary large finite cardinality, i.e., a complete subgraph with all edges of the same color.} 

\

For suitable explications of what ``large set'' means we obtain various Ramsey properties. For example, ``large set'' may mean a ``set of cardinality at least $f(n)$'', where $f$ is a function from natural numbers to natural numbers on a universe with $n$ elements (see e.g. \cite{HVW97}). We will adopt this interpretation and study the computational complexity of the Ramsey quantifiers determined by various functions $f$. Note that in our setting of finite models with one binary relation~$S$, that we will describe below, Ramsey quantifiers are essentially equivalent to the problem of determining whether a graph has a clique of a certain size.

\subsection{Basic Proportional Ramsey Quantifiers}\label{PropRam}
Let us start with a precise definition of ``large relative to the universe''.
\begin{definition}\label{q-large}
 For any rational number $q$ between $0$ and $1$  we say that {\it the set $A \subseteq U$ is $q$-large relative to $U$} if and only if \[\frac{|A|}{|U|} \geq q.\] 
\end{definition}
In this sense $q$ determines the {\it basic proportional Ramsey quantifier ${\sf R_q}$} 
of type $(2)$ (cf. Definition \ref{defGQs}).

\begin{definition}
Let~$\mathbb{M}= (M, S)$ be a relational model with universe~$M$ and a binary relation~$S$.  We say that $\mathbb{M} \in {\sf R_q}$  iff there is a $q$-large (relative to $M$) $A \subseteq M$ such that for all $a, b \in A$, $\mathbb{M} \models S(a, b)$.
\end{definition}


\begin{theorem}[\cite{Szy2011}]
\label{q-clique}
For every rational number $q$, such that $0<q<1$, the corresponding Ramsey quantifier ${\sf R_q}$ is $\NP$-complete. 
\end{theorem}

\subsection{Tractable Ramsey Quantifiers}\label{TRQ}

We have shown an example of a family of NP-complete Ramsey quantifiers. In this section we will describe a class of Ramsey quantifiers that are computable in polynomial time. 
Let us start with considering an arbitrary function $f: \mathbb{N} \rightarrow \mathbb{N}$.

\begin{definition}
 We say that a {\it set $A\subseteq U$ is $f$-large relatively to $U$} iff \[|A| \geq f(|U|).\]
\end{definition}

Then we define Ramsey quantifiers  of type $(2)$ corresponding to the notion of ``$f$-large''.

\begin{definition}
  We define {\it ${\sf R_f}$} as the class of relational models  $\mathbb{M} = (M,S)$, with universe~$M$ and a binary relation~$S$, such that there is an $f$-large set $A \subseteq M$ such that for each $a, b \in A$, $\mathbb{M} \models S(a, b)$.
\end{definition}

Notice that the above definition is very general and covers all previously defined Ramsey quantifiers. For example, we can reformulate Theorem \ref{q-clique} in the following way:

\begin{corollary}
Let $f(n)= \lceil rn \rceil $, for some rational number $r$ such that $0 < r < 1$.
Then the quantifier ${\sf R_f}$ defines a $\NP$-complete class of finite models. 
\end{corollary}

Let us put some further restrictions on the class of functions we are interested in. First of all, as we will consider $f$-large subsets of the universe we can assume that for all $n \in \mathbb{N}$, $f(n) \leq n+1$. In that setting the quantifier ${\sf R_f}$ says about a set $A$ that it has at least $f(n)$ elements, where $n$ is the cardinality  of the universe. We allow the function to be equal to $n+1$ just for technical reasons as in this case the corresponding quantifier has to be always false.

Our crucial notion goes back to paper \cite{Van97} of V\"{a}\"{a}n\"{a}nen:

\begin{definition}\label{boundedness} 
We say that a function {\it $f$ is bounded} if there exists a positive integer~$m$ such that for all positive integers~$n$ it holds that:
\[ f(n) < m \qquad\mtext{ or }\qquad n - m < f(n).\]
Otherwise, $f$ is {\it unbounded}.
\end{definition}

%
%
%
%
%
%
%
%
%
%
%
%


\begin{theorem}[\cite{Szy2011}]
\label{PRam}
If  $f$ is polynomial-time computable and bounded, then the Ramsey quantifier ${\sf R_f}$ is polynomial-time computable.  
\end{theorem}
\begin{proof}[sketch]
Let~$m$ be the integer such that for all~$n$ it holds
that either~$f(n) < m$ or~$n-m < f(n)$.
This means that for every model~$\mathbb{M} = (M,S)$
with~$|M| = n$, to decide if~$\mathbb{M} \in {\sf R_f}$,
we only need to consider those subsets~$A \subseteq M$
for which holds~$|A| < m$ or~$|A| > n-m$.
Since~$m$ is a constant, these are only polynomially many.
\end{proof}

\subsection{Intermediate Ramsey Quantifiers}
We have shown that proportional Ramsey quantifiers define NP-complete classes of finite models. On the other hand, we also observed that bounded Ramsey quantifiers are polynomial-time computable.

The first question we might ask is whether for all functions $f$ the Ramsey quantifier ${\sf R_f}$  is either polynomial-time computable or $\NP$-complete. We observe that this cannot be the case if we make some standard complexity-theoretic assumptions. 

\begin{theorem}
Let~$f(n) = \lceil \log n \rceil$. The quantifier ${\sf R_f}$ is neither polynomial-time computable nor $\NP$-complete, unless the ETH fails.
\end{theorem}
\begin{proof}
Firstly assume that ${\sf R_f}$ is NP-complete.
This means that there is a polynomial-time reduction~$R$
from~$\textsc{3-SAT}$ to  ${\sf R_f}$
(that takes as input an instance of~$\textsc{3-SAT}$ with~$n$ variables
and produces an equivalent instance of ${\sf R_f}$ with~$n' = n^d$ elements,
for some constant~$d$).
There is a straightforward brute force search algorithm~$A$
that solves ${\sf R_f}$ in time~$O((n')^{f(n')}) = O((n')^{\lceil \log n' \rceil})$.
Composing~$R$ and~$A$ then leads to an algorithm
that solves~$\textsc{3-SAT}$ in
time~$O((n^d)^{\lceil\log n^d\rceil}) = O(n^{d^2 \log n})
= O(2^{d^2 (\log n)^2})$, for some constant~$d$,
which runs in subexponential time.
Therefore, the ETH fails.

On the other hand, it is known that if the problem of deciding whether
a given graph with $n$ vertices has a clique of size~$\geq \log n$
(equivalently ${\sf R_f}$, for~$f(n) = \lceil \log n \rceil$)
is solvable in polynomial time,
then the ETH fails \cite[Theorem~3.4]{Cai:2002zl}.
\end{proof}

In other words, assuming the ETH, there exist Ramsey quantifiers whose model checking problem is an example of an $\NP$-intermediate problem in computational complexity, i.e., it is a problem that is in $\NP$ but is neither polynomial-time computable nor $\NP$-complete~\cite{Ladner:1975ty}.

The remaining open question is whether there exists a natural (and broad) class of functions for which we can identify exactly which Ramsey quantifiers (based on functions in this class) are polynomial-time computable,
assuming a reasonable complexity assumption (such as the ETH).
In other words:

\begin{problem}
Can we distinguish a natural and broad class~$\FFF$ of functions for which we can identify a dichotomy theorem (for a reasonable complexity-theoretic assumption~$A$) of the following form: assuming~$A$, we can effectively characterize for each function~$f \in \FFF$ whether the Ramsey quantifier corresponding to~$f$ is polynomial-time solvable or not?%
\footnote{Any sensible dichotomy theorem needs to provide a way of
determining for functions~$f \in \FFF$ whether the Ramsey quantifier~${\sf R_f}$
corresponding to~$f$ is polynomial-time solvable or not,
based on properties of~$f$ (other than the property of
whether~${\sf R_f}$ is polynomial-time computable or not).
It follows trivially (from the law of the excluded middle)
that each problem is either polynomial-time solvable
or not polynomial-time solvable.}
\end{problem}

This generalizes traditional dichotomy theorems
(such as Schaefer's Theorem~\cite{Schaefer:1978rz}) in the following way.
These theorems state that for a set~$\QQQ \subseteq \NP$ of decision problems,
each problem~$Q \in \QQQ$ is either polynomial-time solvable
or it is \NP{}-complete, and it provides a way of determining for any
particular problem~$Q \in \QQQ$ which is the case.
The complexity-theoretic assumption underlying the dichotomy
of the problems in~$\QQQ$ in this case is the assumption that~$\P \neq \NP$.
Since we allow arbitrary complexity-theoretic assumptions~$A$,
our schema is a generalization.

In the remainder of the paper, we provide such a dichotomy result
for the Ramsey quantifiers based on a broad class of polynomial-time and nondecreasing functions~$\FFF$,
where the underlying complexity-theoretic assumption is the ETH.
The technical results underlying our dichotomy originate in the literature
on the (parameterized) complexity of the $k$-\textsc{clique} problem
\cite{Chen:2005fv,ChenHuangKanjXia06}.

\section{More Tractable Ramsey Quantifiers}
In order to show our dichotomy result,
we begin by extending the class of tractable Ramsey quantifiers.
We show that there are functions~$f$ that are not bounded,
but for which~${\sf R_f}$ is polynomial-time computable.
Consider the function~$f(n) = n - c \lceil \log n \rceil$,
where~$c$ is some fixed constant.
Clearly, this function~$f$ is not bounded
(in the sense of Definition~\ref{boundedness}).
We show that for functions of this kind,
the Ramsey quantifier~${\sf R_f}$ is polynomial-time computable.

\begin{proposition}
\label{prop:n-min-log-n}
Let~$c \in \mathbb{N}$ be a constant, and
let~$f : \mathbb{N} \rightarrow \mathbb{N}$ be any polynomial-time computable
function such that, for sufficiently large~$n$, $f(n) \geq n - c \log n$.
Then~${\sf R_f}$ is polynomial-time computable.
\end{proposition}
\begin{proof}
Firstly, we consider the problem of, given a simple graph~$G = (V,E)$
with~$n$ vertices, and an integer~$k$, deciding whether~$G$ contains
a clique of size at least~$n-k$.
We know that this problem can be solved in time~$2^k \cdot \mtext{poly}(n)$
\cite[Proposition~4.4]{Flum:2006fk}.
In other words, deciding whether a graph with~$n$ vertices contains a
clique of size~$\ell$ can be done in time~$2^{n-\ell} \cdot \mtext{poly}(n)$.
We will use this result to show polynomial-time computability of~${\sf R_f}$.

Let~$\mathbb{M}$ be a structure with a universe~$M$
containing~$n$ elements,
and let~${\sf R_f} xy\ \varphi(x,y)$ be an ${\sf R_f}$-quantified
formula.
We construct the graph~$G = (V,E)$ as follows.
We let~$V = M$, and for each~$a,b \in M$ we
let~$E$ contain an edge between~$a$ and~$b$ if and only
if~$\mathbb{M} \models \varphi(a,b)$.
Clearly,~$G$ can be constructed in polynomial time.

Moreover,~$G$ has a clique of size~$f(n)$ if and only
if~$\mathbb{M} \models {\sf R_f} xy\ \varphi(x,y)$.
Therefore, it suffices to decide whether~$G$ has a clique of size~$f(n)$.
We know that~$f(n) \geq n - c \log n$.
As mentioned above, we know we can decide this in
time~$2^{n-f(n)} \cdot \mtext{poly}(n)$.
Because~$n-f(n) \leq c \log n$, we get that~$2^{n-f(n)} \leq 2^{c \log n} = n^c$.
Thus, we can solve the problem in polynomial time.
\end{proof}

The above result can be nicely phrased using a notion of boundedness
that differs from the one in Definition~\ref{boundedness}.

\begin{definition}[Constant-log-boundedness]
Let~$f : \mathbb{N} \rightarrow \mathbb{N}$ be a computable function.
We say that~$f$ is \emph{constant-log-bounded} if
there exists a constant~$c \in \mathbb{N}$ such that for
all sufficiently large~$n \in \mathbb{N}$
(i.e., all~$n \geq n_0$ for some fixed~$n_0 \in \mathbb{N}$),
one of the following holds:
\begin{itemize}
  \item $f(n)$ is bounded above by the constant~$c$,
    i.e., it holds that~$f(n) \leq c$; or
  \item $f(n)$~differs from~$n$ by at most~$c \log n$,
    i.e., it holds that~$f(n) \geq n - c \log n$.
\end{itemize}
\end{definition}

\begin{corollary}
\label{cor:constant-log-bounded}
Let~$f : \mathbb{N} \rightarrow \mathbb{N}$ be a constant-log-bounded
function.
Then~${\sf R_f}$ is polynomial-time computable.
\end{corollary}
\begin{proof}
The algorithms in the proofs of Theorem~\ref{PRam}
and Proposition~\ref{prop:n-min-log-n} can straightforwardly
be combined to work for all constant-log-bounded functions~$f$
-- that is, also for those functions that are appropriately lower bounded
or upper bounded for all~$n \in \mathbb{N}$, but neither
appropriately lower bounded for all~$n$ nor
appropriately upper bounded for all~$n$.
\end{proof}

\section{Intractable Ramsey Quantifiers}

In this section, we show for a large natural class of natural functions~$f$
that~${\sf R_f}$ is not polynomial-time computable, unless the ETH fails.

\subsection{Restrictions on the Class of Functions}

One way in which we assume the functions~$f$ to be natural
is that the value~$f(n)$ is computable in time polynomial in~$n$.%
\footnote{In other words, technically,
we assume that~$f(n)$ is polynomial-time
computable when the value~$n$ is given in unary}
From now on, we will assume that this property holds
for all functions~$f$ that we consider. This assumption corresponds
to restricting the attention to polynomial-time computable monadic generalized
quantifiers which seems reasonable from a natural language
perspective \cite{Szy2011}. 

In fact, for any function~$f$ that is not polynomial-time computable,
the problem~${\sf R_f}$ cannot be computable in polynomial time
either.
\begin{proposition}
\label{prop:ptime}
Let~$f : \mathbb{N} \rightarrow \mathbb{N}$ be a function that is
not polynomial-time computable.
Then~${\sf R_f}$ is not polynomial-time computable.
\end{proposition}
\begin{proof}
Let~$f$ be a function that is not polynomial-time computable
and suppose (to derive a contradiction) that~${\sf R_f}$
is polynomial-time computable.
We give a polynomial-time algorithm to compute~$f$.

Let~$n$ be an arbitrary positive integer.
Consider the family~$(G_i)_{i=1}^{n}$ of graphs,
where for each~$1 \leq i \leq n$, the graph~$G_i = (V,E_{i})$
has a vertex set~$V = \SBs 1,\dotsc,n \SEs$ of size~$n$,
and where:
\[ E_{i} = \SB \SBs j_1,j_2 \SEs \SM 1 \leq j_1 < j_2 \leq i \SE. \]
That is, each graph~$G_i$ has~$n$ vertices, and the edges form
a clique on the first~$i$ vertices (and there are no other edges).

Now, by definition of~${\sf R_f}$, it follows that~$G_i \in {\sf R_f}$
if and only if~$i \geq f(n)$.
Therefore, to compute the value~$f(n)$,
one only needs to identify the maximum value~$i$
such that~$G_i \in {\sf R_f}$.
Since there are only~$n$ graphs~$G_i$,
and deciding whether~$G_i \in {\sf R_f}$ can be done in polynomial time,
the value~$f(n)$ can be computed in polynomial time.

This is a contradiction with our assumption that~$f$ is not polynomial-time
computable. Therefore, we can conclude that~${\sf R_f}$
is not polynomial-time computable.
\end{proof}

Considering this result, in the remainder of the paper we will only
consider functions~$f$ that are polynomial-time computable.

\begin{assumption}
The functions~$f$ that we consider for our dichotomy result
are polynomial-time computable,
i.e., the value~$f(n)$ is computable in time polynomial in~$n$.
\end{assumption}

Moreover, to smoothen the technical details,
we will also direct our attention to nondecreasing functions.

\begin{assumption}
The functions~$f$ that we consider for our dichotomy result
are nondecreasing,
i.e., for all~$n \in \mathbb{N}$ it holds that~$f(n) \leq f(n+1)$.
\end{assumption}


\subsection{Intractability Based on the ETH}

In this section,
we set out to prove the technical results
that will give us the dichotomy result for~${\sf R_f}$,
for the class of polynomial-time computable functions~$f$.
We start with considering the following class of
sublinear functions, that we will use in order to
prove the dichotomy result.

\begin{definition}[Sublinear functions]
Let~$f : \mathbb{N} \rightarrow \mathbb{N}$
be a nondecreasing function.
We say that~$f$ is \emph{sublinear}
if~$f(n)$ is~$o(n)$,
i.e., if there exists some computable function~$s(n)$
that is nondecreasing and unbounded,
and some~$n_0 \in \mathbb{N}$,
such that for all~$n \in \mathbb{N}$ with~$n \geq n_0$
it holds that~$f(n) \leq \frac{n}{s(n)}$.
\end{definition}

In order to illustrate this concept,
we give a few examples of
sublinear functions.

\begin{example}
Consider the function~$f_1(n) = \lceil \log n \rceil$.
This function is sublinear,
which is witnessed by~$s(n) = n/\lceil \log n \rceil$.
Additionally, any function~$f(n)$ that satisfies
that~$f(n) \leq \lceil \log n \rceil$, for all~$n \in \mathbb{N}$,
is also sublinear.
Next, the function~$f_2(n) = \lceil \sqrt{n} \rceil$ is also sublinear,
which is witnessed by~$s(n) = \sqrt{n}/2$.
As a final example, consider the function~$f_3(n) = \lceil n / \log n \rceil$.
Clearly, by taking~$s(n) = \log n / 2$,
we get that~$f_3(n) \leq n/s(n)$.
Therefore,~$f_3$ is also sublinear.
\end{example}

Moreover, we consider the class of functions whose growth
is lower bounded by a linear function.

\begin{definition}[Linearly lower bounded functions]
Let~$f : \mathbb{N} \rightarrow \mathbb{N}$
be a nondecreasing function.
We say that \emph{the growth of~$f$ is
(asymptotically) lower bounded by a linear function}
if~$f(n)$ is~$\Omega(n)$,
i.e., if there exists some constant~$c$
and some~$n_0 \in \mathbb{N}$ such that for
all~$n \geq n_0$ it holds that~$f(n) \geq cn$.
\end{definition}

In order to illustrate this concept,
we give an example of a linearly lower bounded
function.

\begin{example}
Consider the function~$f_4(n) = n - (\log n)^2$.
Because~$(\log n)^2 \leq \frac{n}{2}$ for all~$n \geq 80$,
we get that~$f_4(n) \geq \frac{n}{2}$ for all~$n \geq 80$.
Therefore, the growth of~$f_4$ is asymptotically bounded
by a linear function.
\end{example}


Now, we show that any nondecreasing, computable
function~$f : \mathbb{N} \rightarrow \mathbb{N}$
fits one of four cases.
This will allow us to consider these cases separately,
in order to prove our dichotomy result.
The first case covers all bounded functions;
the second case covers all unbounded but sublinear functions;
and the third and fourth case cover all functions that are
(asymptotically) lower bounded by a linear function.
The difference between the third and fourth case is that the
fourth case covers all functions that differ from~$n$
by at most a factor~$c \log n$, whereas the third case covers
functions whose difference with~$n$ is higher.

\begin{lemma}
\label{lem:case-distinction}
Let~$f : \mathbb{N} \rightarrow \mathbb{N}$ be a nondecreasing,
computable function.
Then one of the following is the case:
\begin{enumerate}
  \item $f(n)$ is~$O(1)$, i.e., there is some constant~$c$ such
    that for all~$n \in \mathbb{N}$,~$f(n) \leq c$;
  \item $f(n)$ is unbounded and~$f(n)$ is~$o(n)$;
  \item $f(n)$ is~$\Omega(n)$ and for sufficiently large~$n$ it holds
    that~$f(n) \leq n - s(n) \log n$, for some unbounded, nondecreasing,
    computable function~$s$; or
  \item for sufficiently large~$n$,~$f(n) \geq n - c \log n$, for some
    constant~$c$.
\end{enumerate}
\end{lemma}
\begin{proof}
There are only two possibilities that have to be ruled out to show that
these four cases are exhaustive.
In the first such possibility, the function~$f$ is (i)~not~$O(1)$, i.e.,%
~unbounded, (ii)~not~$o(n)$,
and (iii)~not~$\Omega(n)$.
We show that these assumptions lead to a contradiction.
This can only be the case if for each nondecreasing,
unbounded, computable function~$s$ it holds that there
is no~$n_s \in \mathbb{N}$ such that for all~$n \geq n_s$
it holds that~$f(n) \leq \frac{n}{s(n)}$.
Then, consider the function~$s'(n) = \frac{n}{f(n)}$.
Since~$f$ is computable, so is~$s'$.
Moreover, since for sufficiently large values of~$n$,%
~$f(n) < n$, we know that~$s'$ is unbounded.
We then construct the function~$s''$ by letting:
\[ s''(n) = \begin{dcases*}
  s'(n) & if~$n = 0$; \\
  s'(n) & if~$n > 0$ and~$s'(n) \geq s''(n-1)$; \\
  s''(n-1) & if~$n > 0$ and~$s'(n) < s''(n-1)$. \\
\end{dcases*} \]
We have that for all~$n \in \mathbb{N}$ it holds that~$s''(n) \geq s'(n)$,
and~$s''$ is unbounded, nondecreasing, and computable.
Now, finally, consider the function~$f'(n) = \frac{n}{s''(n)}$.
Since for all~$n \in \mathbb{N}$ it holds that~$s''(n) \geq s'(n)$,
we know that for all~$n$ it also holds that~$f'(n) \leq f(n)$.
Therefore, it must hold that for each nondecreasing, unbounded,
computable function~$s$ there is no~$n_s \in \mathbb{N}$
such that for all~$n \geq n_s$ it holds that~$f'(n) \leq \frac{n}{s(n)}$.
However, we know that for all~$n \geq 0$ it holds
that~$f'(n) = \frac{n}{s''(n)}$, which is a contradiction.
Therefore, we can rule out this first possibility.

In the second possibility that we have to rule out, (i)~the function~$f$
is~$\Omega(n)$, (ii)~it does not hold that, for sufficiently large~$n$,%
~$f(n) \geq n - c \log n$ for some fixed constant~$c$,
and (iii)~it does not hold that, for sufficiently large~$n$,%
~$f(n) \leq n - s(n) \log n$, for some unbounded, nondecreasing,
computable function~$s : \mathbb{N} \rightarrow \mathbb{N}$.
We show that these assumptions lead to a contradiction.
This can only be the case if for each nondecreasing,
unbounded, computable function~$s$ it holds that there
is no~$n_s \in \mathbb{N}$ such that for all~$n \geq n_s$,%
~$f(n) \leq n - s(n) \log n$.
Then, consider the function~$t'(n) = \frac{n-f(n)}{\log n}$.
Since~$f(n)$ is computable, so is~$t'$.
Moreover, since for sufficiently large values of~$n$,%
~$f(n) < n-\log n$, we know that~$t'$ is unbounded.
We then construct the function~$t''$ similarly to the function~$s''$
defined above (replacing~$s'$ by~$t'$ in the definition).
We then have that for all~$n \in \mathbb{N}$ it holds that~$t''(n) \geq t'(n)$,
and~$t''$ is unbounded, nondecreasing, and computable.
Now, finally, consider the function~$f'(n) = n - t''(n) \log n$.
Since for all~$n \in \mathbb{N}$ it holds that~$t''(n) \geq t'(n)$,
we know that for all~$n$ it also holds that~$f'(n) \leq f(n)$.
Therefore, it must hold that for each nondecreasing, unbounded,
computable function~$s$ there is no~$n_s \in \mathbb{N}$
such that for all~$n \geq n_s$ it holds that~$f'(n) \leq n - s(n) \log n$.
However, we know that for all~$n \geq 0$ it holds
that~$f'(n) = n - t''(n) \log n$, which is a contradiction.
Therefore, we can rule out the second possibility.
\end{proof}

\subsubsection{Sublinear functions}
We turn our attention to sublinear functions~$f$.
We begin with proving a technical lemma.

\begin{lemma}
\label{lem:small-f}
Let~$f : \mathbb{N} \rightarrow \mathbb{N}$
be a nondecreasing function
that is~$o(n)$,
and let~$b \in \mathbb{N}$ be a positive integer.
Moreover, let~$G = (V,E)$ be an instance of~${\sf R_f}$.
In polynomial time,
we can produce some~$b' \geq b$ and we can
transform~$G$ into an equivalent instance~$G' = (V',E')$
of ${\sf R_f}$ with~$n'$ vertices such that~$f(n') \leq b'$.
\end{lemma}
\begin{proof}
Let~$s$ be the unbounded, nondecreasing, computable function
such that~$f(n) \leq n/s(n)$ for sufficiently large~$n$.
If~$f(n) \leq b$, we can let~$G' = G$.
Therefore, assume that~$f(n) > b$.
We will increase~$n$ and~$b$,
by adding a polynomial number of `dummy' vertices
that are connected to all other vertices
(and increasing~$b$ by an equal amount).
It is straightforward to see that such a transformation
results in an equivalent instance.
Since~$s$ is nondecreasing and unbounded, we know there
exists some~$n_0 \in \mathbb{N}$ such that for all~$n \geq n_0$
it holds that~$s(n) \geq 2$.
Now, we define the function~$\delta(n) = n + n_0$.
Clearly,~$\delta$ is polynomial-time computable.
We show that for all~$n,b \in \mathbb{N}$ it holds
that~$f(n + \delta(n)) \leq b + \delta(n)$:
\[ \begin{array}{r l}
  f(n + \delta(n)) & = f(2n + n_0) \leq \frac{2n + n_0}{s(2n + n_0)}
\leq \frac{2n + n_0}{2} = n + \frac{n_0}{2} \\[5pt]
 & \leq n + n_0 \leq
b + n + n_0 = b + \delta(n). \\
\end{array} \]
Now, let~$b' = b + \delta(n)$.
Then, if we add~$\delta(n)$ many vertices to~$G$
that are connected to all other vertices,
we get an instance~$G' = (V',E')$ with~$n'=n + \delta(n)$ vertices
such that~$f(n') \leq b'$.
\end{proof}

Using this lemma, we can now show that
for any nondecreasing, unbounded, computable function~$f$
that is sublinear, the Ramsey quantifier~${\sf R_f}$ is
intractable (unless the ETH fails).

\begin{proposition}
\label{prop:fclique-lower-bound1}
Let~$f : \mathbb{N} \rightarrow \mathbb{N}$
be a nondecreasing, unbounded, computable function
that is~$o(n)$.
Then ${\sf R_f}$ is not solvable in polynomial time,
unless the ETH fails.
\end{proposition}
\begin{proof}
In order to prove our result, we will assume that
${\sf R_f}$ is solvable in polynomial time,
and then show that the ETH fails.
In particular, we will show that \kclique{} is solvable
in time~$f'(k)m^{o(k)}$, which
implies the failure of the ETH by Theorem~\ref{thm:kclique-eth}.

Firstly, we will define a function~$f^{-1}$ as follows.
We let:
\[ f^{-1}(h) = \min \SB q \SM f(q) \geq h \SE. \]
Since~$f$ is an unbounded nondecreasing function,
we get that~$f^{-1}$ is an unbounded nondecreasing function as well.

We give an algorithm that solves \kclique{}
in the required amount of time.
Let~$(G,k)$ be an instance of \kclique{}, where~$G = (V,E)$
is a graph with~$n$ vertices.
Let~$m$ denote the size of~$G$ (in bits).
Intuitively, we will add exactly the right number of `dummy' vertices to~$G$,
resulting in a graph~$G' = (V',E')$,
to make sure that~$f(n') = k$ where~$n' = |V'|$
(while ensuring that~$G$ has a $k$-clique if and only
if~$G'$ has a $k$-clique).
To be more precise, we will construct a number~$k'$
such that~$f(n') = k'$ and such that~$G$ has a $k$-clique
if and only if~$G'$ has a $k'$-clique.
Consider the number~$q = f^{-1}(k)$,
and define~$\ell = f(q) - k$.
By definition of~$f^{-1}$, we know
that~$f(f^{-1}(k)) \geq k$, and thus that~$\ell \geq 0$.
We may assume without loss of generality that~$\ell \leq q-n$
and thus that~$0 \leq \ell \leq q-n$.
If this were not the case, we could invoke Lemma~\ref{lem:small-f}
(by taking~$b = q-n$)
to increase~$q$ to a number~$q'$ (and update~$\ell$ to~$\ell'$
accordingly) such that~$q'-n \geq q-n \geq f(q') \geq f(q') - k = \ell'$.

We now construct~$G'$ from~$G$
by adding~$q-n$ many new vertices,
where~$\ell$ of them are connected in~$G'$
to all existing vertices in~$G$,
and the remaining new vertices are not connected to any other vertex.
We then get that~$n' = q$,
and we let~$k' = f(n') = f(q)$.
It is now straightforward to verify that~$G$ has a $k$-clique
if and only if~$G'$ has a $k'$-clique,
and that the size of~$G'$ is at most~$f^{-1}(k)m^c$
for some constant~$c$. 

Now that we constructed~$G'$, we can use our polynomial-time
algorithm to check whether~$G' \in {\sf R_f}$,
which is the case if and only if~$(G,k) \in \kclique{}$.
This takes an amount of time that is polynomial in the size~$m'$
of~$G'$.
Since~$m' \leq f^{-1}(k)m^c$ for some constant~$c$,
the combined algorithm of producing~$G'$
and deciding whether~$G' \in {\sf R_f}$
takes time~$f'(k)m^{d}$ for some function~$f'$
and some constant~$d$.
From this we can conclude that \kclique{}
is solvable in time~$f'(k) m^{d} = f'(k) m^{o(k)}$.
Therefore, by Theorem~\ref{thm:kclique-eth},
the ETH fails.
\end{proof}

We point out that the result of
Proposition~\ref{prop:fclique-lower-bound1}
actually already follows from a known result
\cite[Theorem~5.5]{ChenHuangKanjXia06}.
For the sake of clarity, we included a self-contained proof
of this statement anyway.


The class of sublinear functions
as considered in the result of Proposition~\ref{prop:fclique-lower-bound1},
also contains those functions~$f$ such that~$f(n) \leq n^{\epsilon}$,
for some constant~$\epsilon$ such that~$0 < \epsilon < 1$.

%

\begin{corollary}
Let~$f : \mathbb{N} \rightarrow \mathbb{N}$ be a unbounded,
computable function such that for
all~$n \in \mathbb{N}$,~$f(n) \leq n^{\epsilon}$ for some
constant rational number~$\epsilon$ such that~$0 < \epsilon < 1$.
Then~${\sf R_f}$ is not polynomial-time computable, unless the ETH fails.
\end{corollary}
\begin{proof}
Since~$f(n) \leq n^{\epsilon}$, we know
that~$f(n) \leq n / n^{1-\epsilon}$.
Then, because~$s(n) = n^{1-\epsilon}$ is a nondecreasing, unbounded
computable function,
we can apply Proposition~\ref{prop:fclique-lower-bound1}
to obtain the intractability of~${\sf R_f}$.
\end{proof}

\subsubsection{Linearly lower bounded functions that are not constant-log-bounded}
Next, we turn to another class of polynomial-time computable functions~$f$
for which~${\sf R_f}$ is not polynomial-time computable unless the ETH
fails.

\begin{proposition}
\label{prop:fclique-lower-bound2}
Let~$f : \mathbb{N} \rightarrow \mathbb{N}$ be a polynomial-time computable
function that is~$\Omega(n)$, and such that, for sufficiently large~$n$,
it holds that~$f(n) \leq n - s(n) \log n$,
for some nondecreasing, unbounded, computable function~$s$.
Then~${\sf R_f}$ is not polynomial-time solvable,
unless the ETH fails.
\end{proposition}
\begin{proof}
We show that a polynomial time algorithm to decide~${\sf R_f}$
can be used to show that deciding whether a given simple graph
(with~$n$ vertices) contains a clique of a given size~$m$
can be solved in subexponential time, i.e., in
time~$2^{o(n)} \mtext{poly}(|G|)$. This, in turn, implies the failure
of the ETH \cite{Impagliazzo:1999uq}.

Let~$G = (V,E)$ be a simple graph with~$n$ vertices. Moreover,
let~$m$ be a positive integer. We will add a certain number,~$\ell$,
of vertices to this graph, to obtain a new graph~$G'$. We will do
this in such a way that almost all of these new vertices
($\ell'$ of them) are connected to all other vertices.
Moreover, we will make sure that~$m + \ell \geq f(n + \ell)$.
Then we can choose~$\ell'$ in such a way
that~$m + \ell' = f(n + \ell)$ --
since~$f(n)$ is~$\Omega(n)$, we can choose~$\ell$
so that $f(n+\ell) \geq m$.
This allows us to use the polynomial time algorithm for~${\sf R_f}$ to
decide whether~$G$ contains a clique of size~$m$, since any clique
of size~$m + \ell'$ in~$G'$ corresponds to a clique of size~$m$
in~$G$.

We define the nondecreasing, unbounded function~$t$
(representing the `inverse' of~$s(n) \log n$) as follows.
Let~$t(n) = \max \SB h \SM s(h) \log h \leq n \SE$.
Since~$s(n) \log n$ grows strictly faster than~$\log n$,
we get that~$t(n)$ is subexponential, i.e.,~$t(n)$ is~$2^{o(n)}$.
Then, in order to ensure that~$m+\ell \geq n+\ell-s(n+\ell)\log(n+\ell)$,
we need that~$s(n+\ell)\log(n+\ell) \geq n-m$,
and thus that~$n+\ell \geq t(n-m)$.
Also, we need to ensure that~$m \leq f(n+\ell)$.
We know that~$f(n) \geq cn$ for some constant~$c$
and for sufficiently large values of~$n$.
We then choose~$\ell = \max \SBs t(n-m)-n, \frac{m}{c} \SEs = 2^{o(n)}$,
which satisfies the required properties.
Therefore, our reduction to~${\sf R_f}$ runs in subexponential-time.
Consequently, if we were to compose this reduction and the
(hypothetical) polynomial time algorithm~${\sf R_f}$, we could decide
whether~$G$ has a clique of size~$m$ in subexponential time,
and thus the ETH fails.
\end{proof}

\subsubsection{The Dichotomy Theorem}
We can now combine the results
of Corollary~\ref{cor:constant-log-bounded},
Lemma~\ref{lem:case-distinction},
Propositions~\ref{prop:n-min-log-n}--\ref{prop:fclique-lower-bound2},
and Theorem~\ref{PRam},
to get the dichotomy result that we were after
for Ramsey quantifiers based on nondecreasing, computable
functions~$f : \mathbb{N} \rightarrow \mathbb{N}$.

\begin{theorem}
Let~$f : \mathbb{N} \rightarrow \mathbb{N}$ be a nondecreasing,
computable function.
Then, assuming the ETH, the Ramsey quantifier~${\sf R_f}$
is polynomial-time computable if and only
if~$f$ is both (1)~polynomial-time computable and (2)~constant-log-bounded.
\end{theorem}
\begin{proof}
The result follows from Corollary~\ref{cor:constant-log-bounded},
Lemma~\ref{lem:case-distinction},
Propositions~\ref{prop:n-min-log-n}--\ref{prop:fclique-lower-bound2},
and Theorem~\ref{PRam}.
Due to Corollary~\ref{cor:constant-log-bounded}, which itself follows
from Proposition~\ref{prop:n-min-log-n} and Theorem~\ref{PRam},
we get that for any polynomial-time computable constant-log-bounded
function~$f$, the quantifier~${\sf R_f}$ is polynomial-time computable.
Conversely, for any function~$f$ that is not polynomial-time computable,
we know by Proposition~\ref{prop:ptime} that~${\sf R_f}$ is not polynomial-time
computable either.
Moreover, for any polynomial-time computable function~$f$ that is
not constant-log-bounded, we know by Lemma~\ref{lem:case-distinction},
that one of two possibilities is the case:
either (i)~$f(n)$ is unbounded and is~$o(n)$,
or (ii)~$f(n)$ is~$\Omega(n)$ and for sufficiently large~$n$ it holds
that~$f(n) \leq n - s(n) \log n$ for some unbounded, nondecreasing,
computable function~$s$.
In case~(i), we know by Proposition~\ref{prop:fclique-lower-bound1}
that~${\sf R_f}$ is not polynomial-time computable, unless the ETH fails.
Similarly, in case~(ii), we know by Proposition~\ref{prop:fclique-lower-bound2}
that~${\sf R_f}$ is not polynomial-time computable, unless the ETH fails.
\end{proof}

\section{Conclusions and Outlook}
We investigated the computational complexity of Ramsey quantifiers. We pointed out some natural tractable (i.e., bounded) and intractable (e.g., proportional) Ramsey quantifiers. These results motivate the search for a dichotomy theorem for Ramsey quantifiers. As a next step, assuming the ETH, we showed that there exist intermediate Ramsey quantifiers (that is, Ramsey quantifiers that are neither polynomial-time computable nor $\NP$-hard). This led to the question whether there exists a natural class of functions, and a notion of boundedness, for which (under reasonable complexity assumptions) the polynomial-time Ramsey quantifiers are exactly the bounded Ramsey quantifiers. We showed that  this is indeed the case. Our main result states that assuming the ETH, a Ramsey quantifier is polynomial-time computable if and only if it corresponds to a polynomial-time computable and constant-log-bounded function. 

An interesting topic for future research is related to identifying the weakest complexity-theoretic assumption that is needed to rule out polynomial-time solvability of~${\sf R_f}$ for different classes functions~$f$. For any function~$f(n) = qn$ for some constant fraction~$q$, the result that~${\sf R_f}$ is not polynomial-time solvable already follows from the weaker assumption that~$\P \neq \NP$. It is not clear what technical obstacles would have to be tackled to identify the exact class of functions~$f$ for which such an intractability result can be established based on the assumption that~$\P \neq \NP$.
It would be interesting to determine whether for the function~$f(n) = n - (\log n)^2$, for instance, it would be possible to prove that~${\sf R_f}$ is not polynomial-time solvable, based on the assumption that~$\P \neq \NP$.

Let us conclude with the following more logically oriented question. The classical property of boundedness plays a crucial role in the definability of polyadic generalized quantifiers. Hella, V\"{a}\"{a}n\"{a}nen, and Westerst{\aa}hl have shown that the Ramseyfication of $\Q$ is definable in $\fo(\Q)$ if and only if $\Q$ is bounded \cite{HVW97}. Moreover, in a similar way, defining ``joint boundedness'' for pairs of quantifiers $\Q_f$ and $\Q_g$ one can notice that $\Br(\Q_f, \Q_g)$ is definable in $\fo(\Q_f, \Q_g)$ \cite{HVW97} and, therefore, polynomial-time computable for polynomial functions $f$ and $g$.
In this paper we substitute the boundedness definition with the notion of constant-log-boundedness, where the bound on the upper side is replaced by $c \log n$. A natural direction for future research is whether this change leads to interesting descriptive results.

\section{Acknowledgments}
Ronald de Haan was supported by the European Research Council (ERC),
project 239962, and the Austrian Science Fund (FWF), project P26200.
Jakub Szymanik acknowledges Veni grant NWO 639.021.232 and the funding received from the European Research Council under the European Union's Seventh Framework Programme (FP/2007--2013)/ERC Grant Agreement n. STG 716230 CoSaQ.

\clearpage
\bibliographystyle{splncs}
\bibliography{jakub}

\end{document}